\documentclass[11pt]{amsart}
\usepackage{amsmath,amssymb,amsfonts,amsbsy}
\usepackage{tikz-cd}
\usetikzlibrary{shapes}
\usepackage{amsbsy}
\usepackage{amscd}
\usepackage{wasysym}
\usetikzlibrary{matrix,arrows,decorations.pathmorphing}
\usepackage{graphicx}
\usepackage{float}
\usepackage{setspace}
\usepackage[margin=1.0in]{geometry}

\newtheorem{thm}{Theorem}
\newtheorem{lem}{Lemma}
\newtheorem{cor}{Corollary}
\newtheorem{prop}{Proposition}

\theoremstyle{definition}
\newtheorem{example}{Example}
\theoremstyle{definition}
\newtheorem{defn}{Definition}
\theoremstyle{definition}
\newtheorem{remark}{Remark}

\newcommand{\rmv}[1]{}

\newcommand{\C}{\mathbb{C}}

\def\ket#1{| #1 \rangle}

\def\kb#1#2{|#1\rangle\!\langle #2 |}

\newcommand{\diag}{\operatorname{diag}}

\begin{document}

\title[Quantum Privacy and Schur Product Channels]{Quantum Privacy and Schur Product Channels}
\author[J. Levick, D.~W. Kribs, R. Pereira]{Jeremy Levick$^{1,2}$, David~W.~Kribs$^{1,3}$, Rajesh Pereira$^{1}$}

\address{$^1$Department of Mathematics \& Statistics, University of Guelph, Guelph, ON, Canada N1G 2W1 dkribs@uoguelph.ca}
\address{$^2$African Institute for Mathematical Sciences, Muizenberg, Cape Town, 7945, South Africa \newline jerlevick@gmail.com}
\address{$^3$Institute for Quantum Computing, University of Waterloo, Waterloo, ON, Canada N2L 3G1 pereirar@uoguelph.ca}

\begin{abstract}
We investigate the quantum privacy properties of an important class of quantum channels, by making use of a connection with Schur product matrix operations and associated correlation matrix structures. For channels implemented by mutually commuting unitaries, which cannot privatise qubits encoded directly into subspaces, we nevertheless identify private algebras and subsystems that can be privatised by the channels. We also obtain further results by combining our analysis with tools from the theory of quasiorthogonal operator algebras and graph theory.
\end{abstract}



\maketitle

\section{Introduction}

Private algebras are fundamental objects of study in the theory of quantum privacy. They are the mathematical manifestation of physically motivated techniques that have been developed to hide qubits from observers in quantum systems in a variety of settings, and also referred to as private quantum channels, private quantum codes and private subsystems \cite{ambainis,bartlett2,bartlett1,boykin,church,cleve,crann,crepeau,jochym,jochym2}. One of the basic challenges in the subject is to identify private codes for specific quantum channels or classes of channels. Only recently has the base of examples begun to expand significantly, starting with a surprisingly simple physical model in \cite{jochym} and expanding to general Pauli channels in \cite{levick}.

In this paper, we begin with the observation that the privacy properties of a large class of quantum channels can be investigated via the matrix theoretic notions of Schur products and correlation matrices. Specifically, all random unitary channels defined by mutually commuting unitaries can be implemented as a Schur product completely positive map with a correlation matrix \cite{kye,li,paulsen}. It is known that these channels cannot privatise qubits encoded directly into subspaces of the underlying Hilbert space for the channel. However, as initially discovered in \cite{jochym}, such channels can nevertheless privatise quantum information, via more delicate subsystem encodings, which from the algebra perspective correspond to more general `ampliated' subalgebras of the full C$^*$-algebra of operators on the system Hilbert space. Here we present a comprehensive analysis of the privacy properties for this class of channels, identifying when private algebras exist for them. Our analysis is based on the underlying correlation matrix graph structure. We then combine this analysis with a connection between operator systems, another graph, and quasiorthonal operator algebra techniques \cite{levick} to obtain further general privatisation results on the channels.

This paper is organised as follows. In the next section we present preliminary notions and a motivating example. In Section~3 we make the connection with Schur products and we use associated correlation matrix and graph structures to identify private algebras and subsystem codes for random unitary channels implemented with mutually commuting operators. In Section~4 we expand our analysis from a different perspective, making use of an operator system and quasiorthogonal algebra approach.

\section{Quantum Channels and Private Algebras}

Let $\Phi:M_n(\C)\rightarrow M_n(\C)$ be a quantum channel: a trace-preserving, completely positive linear map on the $n\times n$ complex matrices. Then
$$\Phi(\rho) = \sum_i A_i \rho A_i^*$$ for some matrices $A_i$, the \emph{Kraus} operators of the channel. That $\Phi$ is trace preserving is equivalent to the condition that $\sum_i A_i^*A_i = I.$

Recall that every $*$-subalgebra $\mathcal{A}$ of $M_n(\C)$ is $*$-isomorphic to an orthogonal direct sum of full matrix algebras $M_{n_1}\oplus \ldots \oplus M_{n_r}$ and is unitarily equivalent to a corresponding direct sum of simple irreducible algebras: $\mathcal A \cong \oplus_{k=1}^r (I_{m_k}\otimes M_{n_k}(\C))$, where $I_m$ is the identity operator on $M_m(\C)$.

\begin{defn}\label{defnpriv} \cite{church,crann,levick}
A channel $\Phi$ \emph{privatises} a $*$-subalgebra $\mathcal{A}$ of $M_n(\C)$, if there exists a fixed density operator $\rho_0$ such that $\Phi(\rho)= \rho_0,$
for all density operators $\rho\in\mathcal{A}$. In many instances that naturally arise we have $\rho_0 = \frac1n I_n$. As this will be our primary focal point in the paper, we shall say the channel {\it privatises to the unit} in this case.
\end{defn}

\begin{remark} Notice that although Definition \ref{defnpriv} defines privatisation of a $*$-subalgebra, the definition of privatisation can be modified in the obvious way to encompass privatisation of any subset of density operators $S$: $\Phi$ privatises $S$ if there exists a fixed $\rho_0$ such that for all $\rho\in S$, $\Phi(\rho)=\rho_0$. Observe that if $\mathcal{A}\simeq M_2(\C)$ then $\Phi$ can privatise a logical qubit of information, and if $\mathcal{A} \simeq M_{2^k}(\C)$ then $k$ qubits can be privatised. We also remark that privatisation to the unit is a seemingly natural type of privatisation to consider, especially when we are interested in channels that privatise $*$-subalgebras. For instance, privatisation was first considered in the context of random unitary channels, where it was noted in \cite{ambainis} that for unital channels, if the privatised set $S$ contains the maximally mixed state, $\rho_0$ must be the identity. For Schur product channels, which is our primary focus, $\Phi$ is unital if and only if $\Phi$ is trace-preserving, so in order to investigate Schur product channels that privatise unital $*$-subalgebras, we must work with privatisation to the unit. Many other interesting and tractable classes of channels, such as random unitary channels are also necessarily unital, and so to study privatisation of $*$-subalgebras in this setting, we again are forced to use privatisation to the unit. It should be possible to extend our analysis to more general privacy for unital channels, using the structure theory for such channels as in \cite{kribsfixed}, but we leave this for consideration elsewhere.
\end{remark}

Here we are primarily interested in studying a class of quantum channels that, at first glance, do not appear to have necessary privatising features. A map $\Phi$ is a {\it random unitary channel} if its Kraus operators are positive multiples of unitary operators; $\Phi(\rho) = \sum_i p_i U_i \rho U_i^*$, with the $p_i$ forming a probability distribution. These channels are centrally important in quantum privacy and quantum error correction, and many privatise quantum information; for instance, the simple complete depolarizing channel, with (unnormalized) Kraus operators given by the Pauli operators $I,X,Y,Z$, satisfies $\Phi(\rho) = \frac12 I$ for all single qubit $\rho$. However, if the unitaries $U_i$ are {\it mutually commuting} one can verify \cite{jochym,jochym2} that $\Phi$ cannot privatise any (unamplified) matrix algebras of the form $M_k(\C)$; that is, $\Phi$ has no private {\it subspaces}. Nevertheless, as we now understand, random unitary channels determined by mutually commuting unitaries can still privatise algebras, but only in the form of so-called private {\it subsystems}, which translates to algebras $I_m\otimes M_n(\C)$ with $m>1$. Consider the following motivating example from \cite{jochym} that illustrates these points, and which we will return to following the analysis of the next section.

\begin{example}\label{nopriv} Let $\Phi: M_2(\C)\rightarrow M_2(\C)$ be the channel whose Kraus operators are $\frac{1}{\sqrt{2}}I$ and $\frac{1}{\sqrt{2}}Z$, where $Z$ is the Pauli $Z$ matrix. Then
$$\Phi(\rho) = \frac{1}{2}\bigl( \rho + Z\rho Z\bigr)$$ acts by projecting $\rho$ down to its diagonal. If $\Phi(\rho) = \frac{\mathrm{tr}}{2}(\rho)I$ then necessarily $\rho$ has constant diagonal. The space of such matrices is of dimension $3$, and so cannot accommodate a qubit $*$-subalgebra. In particular, $\Phi$ does not privatise any qubits.

However, while $\Phi$ itself does not privatise any qubits, the channel $\Phi^{\otimes 2} = \Phi\otimes \Phi$ does privatise a qubit. Indeed, consider the channel $\Phi\otimes \Phi: M_4(\C)\rightarrow M_4(\C)$. The Kraus operators of $\Phi\otimes \Phi$ are $\{\frac{1}{2}I\otimes I,\frac{1}{2} I\otimes Z, \frac{1}{2}Z\otimes I,\frac{1}{2}Z\otimes Z\}$ and so
$$\Phi\otimes \Phi(\rho)= \frac{1}{4}\biggl( \rho + (I\otimes Z)\rho(I\otimes Z)+(Z\otimes I)\rho(Z\otimes I)+(Z\otimes Z)\rho(Z\otimes Z)\biggr),$$
which again is the projection onto the diagonal. Thus, as before, for $\Phi\otimes \Phi(\rho) = \frac{\mathrm{tr}(\rho)}{4}I$ we require that $\rho$ have constant diagonal. But now, there is an algebra isomorphic to $M_2(\C)$ that satisfies this: the algebra generated by
$\{I\otimes X, Y\otimes Z\}$, which is equal to $\mathrm{span}\{ I\otimes I, I\otimes X, Y\otimes Y, Y\otimes Z\}$. Note that $I\otimes X$, $Y\otimes Y$, $Y\otimes Z$ have no non-zero on-diagonal entries, and so the diagonal of $a_1 I\otimes I + a_2 I\otimes X + a_3 Y\otimes Y + a_4 Y\otimes Z$ is $a_1$.
\end{example}

Hence, in addition to the basic fact that random unitary channels with commuting unitaries cannot privatise a subspace (i.e., corresponding to unampliated matrix algebras $M_n(\C)$), we observe that in some cases a tensor power of the channel may yet privatise an ampliated matrix algebra, in this case, an algebra unitarily equivalent to $I_2\otimes M_2(\C)$. These points for this particular class of channels are the focus of our analysis in the next section.

\section{Schur Product Channels and Privacy}

Let $\Phi$ be a channel whose Kraus operators $A_i$ are scalar multiples of mutually commuting unitaries. We begin with an observation that connects private codes for such channels with a fundamental operation in matrix theory. By applying the spectral theorem we may simultaneously diagonalize all the $A_i$; that is, we can find a common basis $\{ \ket{k} \}$ such that the matrix representation for each $A_i$ in this basis is diagonal. We shall write $A_i$ both for the operator and for this diagonal matrix representation, and put $A_i = \diag(A^{(i)}_{k})$ with $A_k^{(i)}\in\mathbb C$ as the diagonal entries of $A_i$.  It follows that the matrix entries in this basis of an output state $\Phi(\rho)= (\Phi(\rho)_{kl})$ satisfy:
\[
\Phi(\rho)_{kl} = \sum_i  \rho_{kl} A^{(i)}_{k} \overline{A^{(i)}_{l}} = \rho_{kl} \sum_i A^{(i)}_{k} \overline{A^{(i)}_{l}}.
\]
Now, if $C = (C_{kl})$ is the matrix given by $$C_{kl} = \sum_i  A^{(i)}_{k} \overline{A^{(i)}_{l}},$$ then we conclude that
$$\Phi(\rho) = C\circ \rho$$ where $\circ$ denotes the entrywise Schur (also called Hadamard) product of two matrices. Moreover, in order for $\Phi$ of this form to be completely positive and trace preserving, necessarily $C$ is positive semidefinite and has all $1$'s down the diagonal. Such a matrix is called a \emph{correlation} matrix. For more on the structure of Schur product channels and correlation matrices, see \cite{kye,li,paulsen}. Let us formulate this observation as a result.

\begin{prop}\label{schurchannel}
Let $\Phi$ be a channel with Kraus operators that are scalar multiples of mutually commuting unitaries. Then there is a correlation matrix $C$ such that, up to a unitary change of basis, $\Phi(\rho) = C \circ \rho$ for all $\rho$.
\end{prop}

Notice that the channel in Example \ref{nopriv} is of the form $(\Phi\otimes\Phi)(\rho) = I_4 \circ \rho$ for all $\rho$, and in particular $\Phi(\rho) = \frac14 I_4$ for all density matrices $\rho$ with constant diagonal. Note that the algebra generated by $\{I\otimes X, Y\otimes Z\}$ consists entirely of matrices with constant diagonal.

Let us consider the Schur product form in more detail from the perspective of private quantum codes. We begin with a simple observation and do not concern ourselves with algebra structures for the moment.

%\begin{prop} Let $\Phi(X) = X\circ C$ be a Schur product map for some correlation matrix $C \in M_n(\C)$. Then $X$ is correctable for $\Phi$ if and only if $c_{ij}=0\Rightarrow x_{ij}$ for all $1\leq i,j\leq n.$
%\end{prop}
%This follows from the obvious observation that we cannot recover $X$ from $\Phi(X)$ only when an entry of %$X$ is multiplied by $0$.
\begin{prop}\label{zeroentries} Suppose that $\Phi$ is a channel and $C$ is a correlation matrix such that $\Phi(\rho)=\rho \circ C$ for all $\rho$. Then a set of density operators $\mathcal S$ is privatised by $\Phi$ to the unit if and only if every $\rho\in\mathcal S$ has constant diagonal and $\rho_{ij}c_{ij}=0$ for all $1\leq i\neq j\leq n$.
\end{prop}

It is clear then, that the zero pattern of $C$ determines what is privatised by $\Phi$; namely, a $*$-subalgebra $\mathcal{A} \subseteq M_n(\C)$ is privatised to the unit for $\Phi$ if and only if for all $\rho\in \mathcal{A}$, $\rho$ has constant diagonal and $\rho_{ij}c_{ij}=0$ for all $i\neq j$.

We next inject into the analysis the graph naturally associated with a correlation matrix.

\begin{defn} Let $C \in M_n(\C)$ be a correlation matrix. The {\it graph of $C$} is the graph $G_C$ whose vertices are $\{1,2,\ldots, n\}$ and whose edge set is $\{(i,j):i\neq j$ and $c_{ij}\neq 0\}$.
\end{defn}

Clearly, $G_C$ contains all information about the location of non-zero entries of $C$. We say that a matrix is ``irreducible'' if its underlying graph is connected.

\begin{prop} Let $C\in M_n(\C)$ be a correlation matrix with graph $G_C$. If $G$ has $m$ connected components $\{G_i\}_{i=1}^m$ each on $k_i$ vertices, then $C$ has $m$ irreducible components $C_i$ corresponding to the $G_i$, and by a permutation may be brought into the form $C=\bigoplus_{i=1}^m C_i$ where each $C_i$ is a $k_i\times k_i$ irreducible correlation matrix.
\end{prop}

The following graph product first introduced in \cite{sabidussi} will be useful in our analysis.

\begin{defn} Let $G,H$ be two graphs. The {\it strong product} of $G$ and $H$, $G\boxtimes H$ has vertex set $V(G)\times V(H)$, and $((i,j),(k,l))\in E(G\boxtimes H)$ if and only if one of the following holds:
\begin{enumerate}
\item $i=k,(j,l)\in E(H)$
\item $j=l,(i,k)\in E(G)$
\item $(i,k)\in E(G),(j,l)\in E(H)$
\end{enumerate}
\end{defn}

Let $C$ be a correlation matrix. Then one can show that the graph of $C\otimes C$ is $G_C\boxtimes G_C$. This follows from the fact that the definition of $G\boxtimes G$ recapitulates exactly when an entry of $C\otimes C$ is non-zero:
\begin{enumerate}
%\item a diagonal entry of $C$ times a diagonal entry is non-zero
\item a diagonal entry of $C$ times a non-diagonal non-zero entry is non-zero
\item a non-diagonal non-zero entry of $C$ times another non-diagonal non-zero entry of $C$ is non-zero.
\end{enumerate}

Observe that if $\Phi$ is a Schur product channel such that $\Phi(\rho)=\rho\circ \, C$ for a correlation matrix $C\in M_n(\C)$, then the tensor product channel satisfies $\Phi^{\otimes k}(\rho) = \rho \circ\, C^{\otimes k}$ for all $k \geq 1$. Hence the privacy features of a tensored Schur product channel $\Phi^{\otimes k}$ can be understood in terms of the non-zero entry structure of $C^{\otimes k}$, which in turn is captured by the graph $G_C^{\boxtimes k}$. Consider the most restrictive case for privacy. Let $K_n$ be the complete graph on $n$ vertices, and note that $K_n^{\boxtimes k} = K_{kn}$. As a direct consequence of Proposition~\ref{zeroentries} we have the following.

\begin{prop} Let $\Phi(\rho)=\rho\circ\, C$ for a correlation matrix $C$ whose graph $G_C$ is the complete graph on $n$ vertices. Then neither $\Phi$ nor any tensor power $\Phi^{\otimes k}$ can privatise a non-scalar subalgebra to the unit.
\end{prop}

More generally, let $G$ be a graph with disjoint connected components $G_i$: $G=\biguplus_i G_i$. Then it follows that
\[
G\boxtimes G = \biguplus_{i,j} (G_i\boxtimes G_j).
\]
In particular, if $G = \biguplus_{i=1}^m K_{k_i}$ is the disjoint union of complete graphs $K_{k_i}$, then we have
%\begin{equation}
\[
G\boxtimes G  = \biguplus_{i,j=1}^m K_{k_ik_j}.
\]
%\end{equation}
Hence, for any correlation matrix $C$ whose graph $G_C$ is the disjoint union of complete graphs, there exists some permutation of $C^{\otimes N}$ such that $C^{\otimes N}$ has the form
\begin{equation}\label{directsum} C^{\otimes N} = \bigoplus_{i_1,\ldots,i_N=1}^m T_{k_{i_1}\ldots k_{i_N}},\end{equation}
where $T_k$ is a $k\times k$ matrix with all non-zero entries.

\begin{lem}\label{submatrix} Let $C$ be a correlation matrix whose graph $G_C$ is the disjoint union of $m$ complete graphs, so that $C^{\otimes N}$ has a direct sum decomposition of the form of Equation \ref{directsum}. Then $C^{\otimes N}$ has a principal submatrix of size $m^N$ that is just the identity matrix.
\end{lem}

\begin{proof}
For each element $I=(i_1,\ldots,i_N)\in \{1,\ldots,m\}^N$ denote by $|I|$ the product $\Pi_{i\in I} k_i$. Also for each such $I$ there is a component of the direct sum of $C$ indexed by $I$: $T_I:=T_{k_{i_1}\ldots k_{i_N}}=T_{|I|}$ appearing as a diagonal block of $C$.

Arrange the elements of $\{1,\ldots,m\}^N$ in lexicographic order so that $I_1 = \{1,1,\ldots,1)$, $I_2 = (1,1,\ldots,2)$, $\cdots, I_{m^N} = (m,m,\ldots,m)$, and permute $C$ so that its $j^{th}$ diagonal block is $T_{I_j}$. Then we can choose the principal submatrix indexed by $J = \{1,|I_1|+1,|I_1|+|I_2|+1,\ldots,\sum_{j=1}^{m^N-1} |I_j| +1 \}.$

Clearly, for any $j\in J$, $C^{\otimes N}_{jj} = 1$, and for $i\neq j\in J$ we have that $C^{\otimes N}_{ij}$ is not in block $I$ for any $I\in \{1,\ldots,m\}^N$. Hence such an entry is equal to $0$.
\end{proof}

This approach yields an alternate proof of Theorem~$2$ from \cite{levick}.

\begin{lem}\label{qubitsub} Let $I$ be the $2^n\times 2^n$ identity matrix. Then the channel $\Phi:M_n(\C)\rightarrow M_n(\C)$ defined by $\Phi(X) = I\circ X$ privatises to the unit a $*$-subalgebra $\mathcal{A}$ unitarily equivalent to $M_{2^{\left \lfloor{n/2}\right \rfloor}}(\C)$.
\end{lem}

Moving beyond this motivating special case, we may utilize the Schur product approach to establish the following general result that applies to all random unitary channels with mutually commuting Kraus operators.

\begin{thm}\label{qubitsprivd} Let $C \in M_n(\C)$ be a correlation matrix whose graph $G_C$ has $m>1$ connected components. Consider the channel $\Phi(\rho) = \rho\, \circ \, C$. Then $\Phi^{\otimes N}$ can privatise to the unit a $*$-subalgebra $\mathcal{A}$ unitarily equivalent to $M_{2^{\left \lfloor{N/2}\right \rfloor}}(\C)$.
\end{thm}

\begin{proof} It suffices to prove the result in the case that all the connected components of $G_C$ are complete graphs, as adding more zeros to a correlation matrix only makes it simpler to privatise an algebra.  Recall that $\Phi^{\otimes N}(Y) = Y\circ C^{\otimes N}$. By Lemma \ref{submatrix} there exists a submatrix $C[J]$ of $C^{\otimes N}$ of size $m^N$ that is just the $m^N\times m^N$ identity matrix. Restricted to this submatrix, $\Phi^{\otimes N}$ is just the map $Y \mapsto I_{m^N}\circ Y$. Since $m\geq 2$, $m^N\geq 2^N$ and so restricting further, we have the map $Y\mapsto I_{2^N}\circ Y$. By Lemma \ref{qubitsub} this map can privatise an algebra $\widehat{\mathcal{A}}\subseteq M_{2^N}$ unitarily equivalent to $M_{2^{\left \lfloor{N/2}\right \rfloor}}(\C)$. Embedding this algebra in the obvious way first into the $m^N\times m^N$ matrices by $\widehat{\mathcal{A}}\oplus I_{m^N-2^N}$ and from there in the $m^N\times m^N$ submatrix indexed by $[J]$, we obtain a subalgebra $\mathcal{A}$ unitarily equivalent to $M_{2^{\left \lfloor{N/2}\right \rfloor
}}(\C)$    such that $\Phi^{\otimes N}$ restricted to this subalgebra is just Schur product with the identity, and so $\mathcal{A}$ is privatised by $\Phi^{\otimes N}$.
\end{proof}

\begin{cor} For a Schur product channel $\Phi(\rho)=\rho\,\circ\, C$ where $C$ has two connected components, $\Phi^{\otimes 2}$ can privatise a qubit.
\end{cor}

This supplies another proof that the channel from Example \ref{nopriv} privatises a qubit when tensored with itself: the channel is $\Phi(\rho) = I_2\circ \rho$ and clearly $G_{I_2}$ has two connected components.  Theorem \ref{qubitsprivd} and its corollary can be compared to the examples given in \cite{shirokov} of a channel $\Phi$ which has zero one-shot quantum zero-error capacity but whose tensor product $\Phi \otimes \Phi$ has positive one-shot quantum zero-error capacity.

\begin{example}
As another example, let $\Phi$ be the two-qubit channel $\Phi(\rho) = \frac12 ( \rho + U \rho U^* )$, where $U$ is the unitary CNOT gate; $U\ket{i}\ket{j} = \ket{i}\ket{i\oplus j}$.  Consider $U$ in its diagonal matrix form $U = \diag (1,1,1,-1)$. Then the corresponding correlation matrix for $\Phi(\rho) = C \, \circ \, \rho$ is
$$C = \begin{bmatrix} 1 & 1 & 1 & 0 \\ 1 & 1 & 1 & 0 \\ 1 & 1 & 1 & 0 \\ 0 & 0 & 0 & 1 \end{bmatrix}.$$
By Theorem~\ref{qubitsprivd}, we know that $\Phi^{\otimes 2}$ privatises a qubit as $G_C$ is made up of two complete components. Let us observe this directly.

It is not hard to see that $\Phi$ is a conditional expectation onto the algebra $M_3\oplus \mathbb{C}$. Thus, $\Phi^{\otimes 2}(\rho) = \rho\,\circ\,(C\otimes C)$, which has correlation matrix permutationally equivalent to $C\otimes C \simeq P_9 \oplus P_3 \oplus P_3 \oplus 1$ where $P_k$ is the $k\times k$ all 1's matrix.  Consider the principal submatrix of $C\otimes C$ supported on the indices 1, 10, 13, 16. For this set of indices, $C_{ii}=1$ obviously, but $C_{ij}=0$ for $i\neq j$ since each index is from a different block and only indices from within the same block will yield a non-zero entry. Let $\mathcal{A}$ be the algebra, which is unitarily equivalent to $M_4\oplus \mathbb{C}I_{12}$, supported on the same indices. Then there is a subalgebra of this which is unitarily equivalent to $\mathrm{Alg}\{I\otimes X, Y\otimes Y, Y\otimes Z\}\oplus \mathbb{C} I_{12}\simeq (I_2\otimes M_2)\oplus \mathbb{C} I_{12}$,  which we know is privatised by the map to the diagonal. However, this algebra has non-trivial support only on the submatrix where $C\otimes C$ is the identity, so $\Phi$ restricted to this algebra is just the algebra $\mathrm{Alg}\{I\otimes X,Y\otimes Y, Y\otimes Z\}$ Schur producted with the $4\times 4$ identity; in other words, it is privatised by $\Phi$.
%Another two tensor products, and you could privatize another qubit, using essentially the same idea.
\end{example}

%\begin{remark}
%[FILL IN. Placeholder for remark on how results can be generalized beyond privatizing to the unit, eg when output state is multiple of a projection.]
%While we have focussed on the case of privatizing to the unit, let us note that the above analysis carries through directly to the more general case of a projection as the fixed output state. Indeed, the fixed point theory for unital quantum channels [cite] can be applied to obtain a variety of details here. Suppose $\Phi$ is a random unitary channel implemented by unitaries $U_i$. If all $\rho$ belonging to some $*$-subalgebra $\mathcal A$ are mapped by $\Phi$ to a multiple of a fixed projection $Q$,
%\end{remark}

We conclude this section by noting there are connections between our analysis and some classical graph theoretical concepts.

\begin{defn} Let $G=(V,E)$ be a graph.  Then an $S\subset V$ is an independent vertex set of $G$ if $(i,j)\not \in E$ for all $i,j\in S$.  The independence number of the graph $G$ (denoted by $\alpha(G)$) is the cardinality of the largest independent vertex set in $G$.   \end{defn}

\begin{prop} Let $C$ be a correlation matrix with graph $G_C$.  Then the largest principal submatrix of $C$ which is an identity matrix is $\alpha(G_C)$ by $\alpha(G_C)$. \end{prop}

If $S$ and $T$ are independent vertex sets of the graphs $G$ and $H$ respectively, then $S\boxtimes T$ is an independent vertex set of $G\boxtimes H$.  It follows from this observation that the independence number is supermultiplicative with respect to the strong product: $\alpha(G\boxtimes H) \ge \alpha(G)\alpha(H)$.

We now introduce the Shannon capacity of a graph, a concept first introduced by Shannon in \cite{shannon} and then extensively studied by Lov{\'a}sz in \cite{lovasz}.

\begin{defn} Let $G$ be a graph.  Then the Shannon capacity of $G$ is denoted as $\Theta(G)$ and is defined as $\Theta(G)=\sup_{n\in\mathbb{N}} \alpha(G^{\boxtimes k})^\frac{1}{k}$. \end{defn}

\begin{remark}
We note that the graph theoretic interpretation of Lemma~\ref{submatrix} is essentially the fact that any disjoint union of $m$ complete graphs has Shannon capacity $m$. From the above analysis, graphs with relatively large Shannon capacity may be useful in quantum privacy. We also note the work \cite{duan}, which makes use of graph theory in quantum information as well, and suggest there are likely connections with quantum privacy worth exploring.
\end{remark}

\section{Privacy, Operator Systems, and Quasiorthgonality}

We begin this section by recalling that an operator system $S$ is a subspace of $M_n(\C)$ that is $*$-closed, and contains the identity $I\in S$. Further, if $\Phi:M_n(\C)\rightarrow M_m(\C)$ is a completely positive map with Kraus operators $\{A_i\}_{i=1}^p$, we define the adjoint channel, $\Phi^{\dagger}$ in terms of the (Hilbert-Schmidt) trace inner product:
\begin{equation} \mathrm{Tr}(X\Phi(Y)) = \mathrm{Tr}(\Phi^{\dagger}(X)Y).\end{equation}
It is straightforward from this definition that
\begin{equation} \Phi^{\dagger}(X) = \sum_{i=1}^p A_i^*X A_i.\end{equation}
If $\Phi$ is trace-preserving, then $\Phi^{\dagger}$ is unital, and $\mathrm{range}(\Phi^{\dagger})$ is an operator system, $S\subseteq M_n(\C)$.

One more idea we need is the following: we say that a $*$-subalgebra $\mathcal{A}\subseteq M_n(\C)$ is {\it quasiorthogonal} to an operator system $S\subseteq M_n(\C)$ if $$n\mathrm{Tr}(sa) = \mathrm{Tr}(s)\mathrm{Tr}(a)$$ for all $s\in S$, $a \in\mathcal{A}$. For more details on quasiorthogonality see \cite{levick} and the references therein. The following result brings these concepts together.

\begin{lem}\label{opsyspriv} Let $\Phi: M_n(\C)\rightarrow M_m(\C)$ be a trace-preserving completely positive map, so that $S = \mathrm{range}(\Phi^{\dagger})$ is an operator system. If $\mathcal{A}$ is a $*$-subalgebra of $M_n(\C)$ quasiorthogonal to $S$, then $\mathcal{A}$ is privatised by $\Phi$.
\end{lem}
\begin{proof} For all $a\in\mathcal{A}$ and $x\in M_n(\C)$ we have
\begin{align*} \mathrm{Tr}(x\Phi(a)) & = \mathrm{Tr}(\Phi^{\dagger}(x)a) \\
& = \frac{1}{n}\mathrm{Tr}(\Phi^{\dagger}(x))\mathrm{Tr}(a) \\
& = \mathrm{Tr}(x\Phi(I)\frac{\mathrm{Tr}(a)}{n})\end{align*}
and so $\Phi(a) = \frac{\mathrm{Tr}(a)}{n}\Phi(I)$.
\end{proof}

Lemma~\ref{opsyspriv} is convenient for us, as it will allow us to state a simple and necessary condition on algebras privatised by a given arbitrary unital channel.

\begin{prop}\label{algsandopsys} Let $\Phi:M_n(\C)\rightarrow M_m(\C)$ be a unital quantum channel, with $S = \mathrm{range}(\Phi^{\dagger})$ the operator system that is the range of the adjoint. Let $\mathcal{A}$ be a $*$-subalgebra contained in $S$.
If a $*$-subalgebra $\mathcal{B}$ is private for $\Phi$, necessarily $\mathcal{A}$ and $\mathcal{B}$ are quasiorthogonal algebras.
\end{prop}
\begin{proof} Apply Lemma \ref{opsyspriv} and the fact that $S$ contains $\mathcal{A}$.
\end{proof}

For the rest of this section we focus on Schur product channels. Observe that in the case that $\Phi(X) = X\circ C$ for some correlation matrix $C$, it is easy to see that $\Phi^{\dagger}(X) = \overline{C}\circ X$, where $\overline{C}$ is the matrix whose entries are the complex conjugates of $C$. Thus, it follows that the range of $\Phi^{\dagger}$ is the operator system spanned by the matrix units $E_{ij}=\kb{i}{j}$ such that $C_{ij} \neq 0$.

\begin{defn} Let $G$ be a graph with vertex set $V = \{1,2,\ldots, n\}$ and edge set $E$. The operator system of $G$, $S_G$, is the subspace of $M_n(\C)$ given by
$$S_G = \mathrm{span}\{E_{ij}: (i,j)\in E(G)\}\cup\{E_{ii}: 1\leq i \leq n\}.$$
\end{defn}
Let $C$ be an $n\times n$ correlation matrix. As above, let $G_C$ be the graph on $n$ vertices with an edge $(i,j)$ whenever $C_{ij} \neq 0$. Then the operator system $\mathrm{range} (X\circ \overline{C})$ is the operator system $S_{G_C}$.\\

If $\Phi(X)=C\circ X$, then $\Phi^{\dagger}(X)=\overline{C}\circ X$, so the operator system $S =\mathrm{range}(\Phi^{\dagger})$ is the same as the operator system $\mathrm{range}(\Phi)$. We can prove a general privacy result for Schur product channels as follows.

\begin{lem}\label{orthogtodiag} If $\mathcal{A}$ is a $*$-subalgebra privatised by a Schur product channel $\Phi:M_n(\C)\rightarrow M_n(\C)$, then $\mathcal{A}$ is quasiorthogonal to the diagonal algebra on $n\times n$ matrices, $\Delta_n$.
\end{lem}
\begin{proof} Note first that for a diagonal matrix $D$, we have $\Phi(D)=C\circ D = D$, and so the full diagonal algebra is contained in the range of $\Phi$, and hence inside $S=\mathrm{range}(\Phi^{\dagger})$.  Given this fact, the result follows as a consequence of Proposition \ref{algsandopsys}.
\end{proof}

Let $\mathcal{A}$ be a $*$-subalgebra of $M_n(\C)$. We recall that a {\it separating vector} for $\mathcal{A}$ is a vector $v\in \mathbb{C}^n$ such that $av = 0$ implies $a=0$ for all $a\in \mathcal{A}$. Given the decomposition of $\mathcal{A}$ up to unitary equivalence, $\mathcal{A} \simeq \bigoplus_{k=1}^m I_{i_k}\otimes M_{j_k}$, one can show that $\mathcal{A}$ admits a separating vector if and only if $i_k \geq j_k$ for all $1\leq k \leq m$. Moreover, it is known \cite{pereira} that
a $*$-subalgebra is quasiorthogonal to $\Delta_n$ if and only if it admits a separating vector. Combining these facts with Lemma~\ref{orthogtodiag} allows us to prove the following theorem.

\begin{thm}\label{privifsepvec} A necessary condition for an algebra $\mathcal{A}$ to be be privatised by a Schur product channel $\Phi:M_n(\C)\rightarrow M_n(\C)$ is that $\mathcal{A}$ admits a separating vector.\end{thm}

We can strengthen Theorem \ref{privifsepvec} in the case that we can find more algebras inside the range of $\Phi$. In order to show how we can obtain stronger results in the case, we first need to recall some notions from matrix and operator theory.

Let $C$ be an $n\times n$ correlation matrix and $\Phi(X)=X\circ C$. Then $\Phi^{\dagger}(X) = \overline{C}\circ X$ is also a unital and trace-preserving quantum channel. Any unital channel has an associated {\it multiplicative domain}; the set
$\{X \in M_n(\C): \Phi(X)\Phi(X^*) = \Phi(XX^*)\},$ which is necessarily an algebra, and $X$ in the multiplicative domain satisfies $\Phi(XA)=\Phi(X)\Phi(A)$ for all $A\in M_n(\C)$. The multiplicative domain is a standard structure of interest in the study of completely positive maps, and more recently it has found applications in quantum error correction \cite{choi,johnston}.

Now let $\mathcal{A}$ be the multiplicative domain of $\Phi^{\dagger}$, then
\begin{equation*} \mathrm{Tr}(XY\Phi(I))  = \mathrm{Tr}(\Phi^{\dagger}(XY))
= \mathrm{Tr}(\Phi^{\dagger}(X)\Phi^{\dagger}(Y))
= \mathrm{Tr}((\Phi\circ\Phi^{\dagger})(X)Y) \end{equation*}
for all $Y$, and so $(\Phi\circ\Phi^{\dagger})(X)=X\Phi(I)$. Thus when $\Phi$ is unital, $X$ is a fixed point of $\Phi\circ\Phi^{\dagger}$. In general the fixed point set of a channel is an operator system, but in the case of a unital channel more is true, it is known to be an algebra \cite{kribsfixed}.
Clearly, the fixed point algebra $\mathrm{Fix}(\Phi\circ\Phi^{\dagger})\subseteq \mathrm{range}(\Phi)$, and so following Proposition~\ref{algsandopsys}, if a $*$-subalgebra $\mathcal{A}$ is privatised by $\Phi$, then $\mathcal{A}$ is necessarily quasiorthogonal to $\mathrm{Fix}(\Phi\circ\Phi^{\dagger})$, and hence also to the multiplicative domain of $\Phi$.

For a Schur product channel, the multiplicative domain algebra is a subset of the fixed point algebra of the map $X\mapsto\overline{C}\circ (C\circ X) = (\overline{C}\circ C)\circ X.$ An operator $X$ is such a fixed point so long as $X_{ij}=0$ for all $(i,j)$ such that $|C_{ij}|^2 \neq 1$. Notice that  Lemma~\ref{orthogtodiag} can be re-characterized in light of the above as the trivial observation that, for a correlation matrix, $|C_{ii}|^2 =1$ for any correlation matrix. For a correlation matrix $C$ with other entries of modulus $1$, we can strengthen the claims of Lemma \ref{orthogtodiag} and Theorem \ref{privifsepvec}; namely, any $*$-subalgebra $\mathcal{A}$ privatised by the channel $\Phi(X) = C\circ X$ must be quasiorthogonal to any algebra which has non-zero entries at $(i,j)$ only if $|C_{ij}|^2 =1$.

\subsection{Unitary Graphs of Correlation Matrices} We conclude by returning to a graph-theoretic perspective and considering the following notion.

\begin{defn} Let $C=(C_{ij})$ be an $n\times n$ correlation matrix. The {\it unitary graph} $UG_C$ of $C$ is the graph whose vertices are labelled by $1,2,\ldots,n$ with an edge from $i$ to $j$ if and only if $|C_{ij}|^2 = 1$.
\end{defn}

\begin{thm} Let $UG_C$ be the unitary graph of a correlation matrix $C$. Then any connected component of $UG_C$ having more than $1$ vertex is a complete graph.
\end{thm}
\begin{proof} We proceed by induction. The first non-trivial case is a connected component with three vertices. Denote the vertices $i,j,k$ with edges $(i,j)$, $(j,k)$. Then the submatrix of $C$ indexed by $i,j,k$ is
\begin{equation*} C[i,j,k] = \begin{pmatrix} 1 & C_{ij} & C_{ik} \\ \overline{C}_{ij} & 1 & C_{jk} \\ \overline{C}_{ik} & \overline{C}_{jk} & 1 \end{pmatrix}\end{equation*} which, as a principle submatrix of a positive semidefinite matrix must be positive itself. Hence
\begin{equation*} \begin{pmatrix} 1 & C_{jk} \\ \overline{C}_{jk} & 1 \end{pmatrix} - \begin{bmatrix}\overline{C}_{ij} \\ \overline{C}_{ik} \end{bmatrix}\begin{bmatrix} C_{ij} & C_{ik} \end{bmatrix} = \begin{pmatrix} 0 & C_{jk} - \overline{C}_{ij}C_{ik} \\ \overline{C}_{jk}-C_{ij}\overline{C}_{ik} & 1-|C_{ik}|^2\end{pmatrix}\geq 0\end{equation*}
and hence $|C_{ik}| = |C_{jk}|/|C_{ij}| = 1$ and $(i,k)$ is an edge as well.

Now, assume that it is true that for all connected components on $n$ vertices, the unitary graph must be complete. Then for a graph of size $n+1$, each subgraph of size $n$ must be complete, and we are done.
\end{proof}

Thus $UG_C$ is a subgraph of $G_C$ consisting of disconnected components $\{G_i\}_{i=1}^m$, each of which is a complete graph on $k_i$ vertices, $\sum_{i=1}^m k_i = n$. So the operator system generated by $UG_C$, $S_{UG_C}$, is (up to unitary equivalence)
\begin{equation*} S_{UG_C} \cong \bigoplus_{i=1}^m M_{k_i}(\C).\end{equation*}
This is, of course, a $*$-subalgebra, and if $\mathcal{A}$ is privatised by $\Phi$, then necessarily $\mathcal{A}$ is quasiorthogonal to the algebra $S_{UG_C}$.

In a sense, this gives us a lower bound on the size of algebras that a private algebra must be quasiorthogonal to. We have already seen that any algebra privatised by a Schur product channel must be quasiorthogonal to the diagonal algebra, and hence must contain a separating vector. Our preceding remarks strengthen this: the algebra defined by having non-zero entries wherever $C$ has an entry of modulus $1$ is a larger algebra, necessarily containing the diagonal algebra, such that any algebra privatised by the channel must be quasiorthogonal to this algebra. This ``lower bound" algebra is, up to unitary equivalence, simply a direct sum of matrix algebras.

%\begin{defn} Given an operator system $S$, define the $*$-algebra generated by $S$ to be the intersection of all algebras containing $S$. Equivalently, it is the algebra generated by any basis for $S$.
%\end{defn}

As a consequence of von Neumann's double commutant theorem, we have that the $*$-algebra generated by $S$ is $S^{\prime\prime}$, the commutant of the commutant of $S$.
Given a graph $G$, the let $S_G^{\prime\prime}$ be the algebra generated by the operator system of $G$, $S_G$. If we define $G^*$ to be the graph with the same vertex set as $G$ and with an edge between two vertices if and only if they are in the same connected component of $G$, then $S_G^{\prime\prime}=S_{G^*}$. It is easy to see that if an algebra $\mathcal{B}$ is quasiorthogonal to $S_G^{\prime\prime}$ then $\mathcal{B}$ is privatised by any Schur product channel $\Phi(X) = X\circ C$ for which the graph of $C$ is $G$.

Hence, given a Schur product channel, $\Phi(X)=X\circ C$, we have two algebras $S_{UG_C}$ and $S_{G_C}^{\prime\prime}$, which define necessary and sufficient conditions respectively for an algebra to be privatised by $\Phi$. Obviously, if these two algebras are equal, which occurs if and only if $C$ is permutationally similar to a direct sum of rank one correlation matrices, we can completely characterize the algebras privatised by $\Phi$.

We conclude with an example that serves to illustrate how these necessary and sufficient conditions are, in general, not strong enough to completely characterize the algebras that are privatised by $\Phi$.

\begin{example} Let $\Phi:M_3(\C)\rightarrow M_3(\C)$ be given by $\Phi(X) = X\circ C$ with
$$C=\begin{pmatrix} 1 & 0 & \frac{1}{2} \\ 0& 1 & \frac{1}{2} \\ \frac{1}{2} & \frac{1}{2} & 1 \end{pmatrix}.$$
The graph of $C$ is
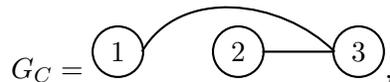
\[ G_C =
\begin{tikzpicture}[thick, scale=0.8]
\node[shape=circle,draw=black](1) at (1,0) {1};
\node[shape=circle,draw=black](2) at (3,0){2};
\node[shape=circle,draw=black](3) at (5,0){3};
\draw(3.4,0) to (4.6,0);
\draw(1.4,0) to[out=60,in=135] (4.6,0);
\end{tikzpicture},\]
while the unitary graph is

\[UG_C = \begin{tikzpicture}[thick, scale=0.8]
\node[shape=circle,draw=black] (1) at (1,0) {1};
\node[shape=circle,draw=black] (2) at (3,0){2};
\node[shape=circle,draw=black] (3) at (5,0){3};
\end{tikzpicture},\]
the completely disconnected graph. Hence, $\mathcal{A}_{G_C} = M_3(\C)$ and $S_{UG_C} = \Delta_3$.

However, being quasiorthogonal to $M_3(\C)$ and $\Delta_3$ are not necessary or sufficient respectively, for an algebra to be privatised by $\Phi$; for instance, the algebra
$$\mathcal{A} = \left\{ \begin{pmatrix} a & b & 0 \\ b & a & 0 \\ 0 & 0 & a \end{pmatrix}: \ a,b\in\C \right\}$$ is privatised by $\Phi$, despite not being quasiorthogonal to $M_3(\C)$, and the circulant algebra
$$\mathcal{C} = \left\{ \begin{pmatrix} a & b & c \\ c & a & b \\ b & c & a \end{pmatrix}: \ a,b,c\in \C\right\}$$ is   quasiorthogonal to $\Delta_3$, but is not privatised by $\Phi$.
\end{example}

\begin{remark}
We conclude with a note on our use of graph theory here and related work. A number of papers have commented on the fact that operator systems behave in many ways like 'non-commutative' or 'quantum' versions of graphs. In \cite{duan} it was shown that certain properties of an operator system can be regarded as quantum analogues of the independence number and Lovasz number of a graph, while in \cite{weaver} quantum cliques have been considered, where an analogue of Ramsey theory for operator systems was established. Other authors have focused on defining quantum chromatic numbers, and other graph parameters \cite{cameron, paulsen2013quantum}. The study of these `quantum' graph parameters is an active and interesting area of research in pure mathematics and quantum information.

What the present analysis aims to do is to understand how certain operator systems associated to a channel control the ability of that channel to privatise information. Schur product channels are the channels whose operator systems are the simplest as they are isomorphic to graphs, and so studying privatisation by means of Schur product channels is the simplest place to start such an analysis. Our work shows that a certain graph parameter, the independence number, controls how copies of a channel can privatise information; we expect that some quantum version of the independence number should play an analogous role for channels whose ranges are operator systems more complicated than those arising directly from graphs.

One final point of note is that not one, but two graphs are important for understanding how a Schur product channel privatises a $\ast$-subalgebra: the graph whose associated operator system is the range of the channel, and a graph that encodes the multiplicative domain of the channel. This seems like a hint that in order to understand how operator systems of channels allow privatisation, the multiplicative domain will be an important object.
\end{remark}

\vspace{0.1in}

{\noindent}{\it Acknowledgements.} J.L. was supported by an AIMS-University of Guelph Postdoctoral Fellowship. D.W.K. was supported by NSERC and a University Research Chair at Guelph. R.P. was supported by NSERC. We are grateful to the referees for helpful comments.

\bibliographystyle{plain}
\bibliography{schur}

\end{document}